\documentclass[12pt]{article}

\def\UseBibLatex{1}

\makeatletter
\def\input@path{{styles/}}
\makeatother

\providecommand{\BibLatexMode}[1]{}
\providecommand{\BibTexMode}[1]{}

\renewcommand{\BibLatexMode}[1]{#1}
\renewcommand{\BibTexMode}[1]{}

\ifx\UseBibLatex\undefined%
  \renewcommand{\BibLatexMode}[1]{}
  \renewcommand{\BibTexMode}[1]{#1}
\fi

\BibLatexMode{%
   \usepackage[bibencoding=utf8,style=alphabetic,backend=biber]{biblatex}%
   \usepackage{sariel_biblatex}%
}

\usepackage[cm]{fullpage}%
\usepackage{amsmath}%
\usepackage{amssymb}%
\usepackage[table]{xcolor}%

\usepackage[amsmath,thmmarks]{ntheorem}%

\usepackage{titlesec}%
\usepackage{xcolor}%
\usepackage{mleftright}%
\usepackage{xspace}%
\usepackage{graphicx}
\usepackage{hyperref}%
\usepackage[inline]{enumitem}
\usepackage{hyperref}%
\usepackage[ocgcolorlinks]{ocgx2}
\usepackage{graphicx}

\hypersetup{%
      unicode,
      breaklinks,%
      colorlinks=true,%
      urlcolor=[rgb]{0.25,0.0,0.0},%
      linkcolor=[rgb]{0.5,0.0,0.0},%
      citecolor=[rgb]{0,0.2,0.445},%
      filecolor=[rgb]{0,0,0.4},
      anchorcolor=[rgb]={0.0,0.1,0.2}%
}

\titlelabel{\thetitle. }%

\theoremseparator{.}%

\theoremstyle{plain}%
\newtheorem{theorem}{Theorem}[section]

\newtheorem{lemma}[theorem]{Lemma}

\theoremstyle{plain}%
\theoremheaderfont{\sf} \theorembodyfont{\upshape}%
\newtheorem*{remark:unnumbered}[theorem]{Remark}%

\newtheorem{defn}[theorem]{Definition}
\newtheorem{example}[theorem]{Example}

\newtheorem{assumption}[theorem]{Assumption}%

\theoremheaderfont{\em}%
\theorembodyfont{\upshape}%
\theoremstyle{nonumberplain}%
\theoremseparator{}%
\theoremsymbol{\myqedsymbol}%
\newtheorem{proof}{Proof:}%

\providecommand{\emphind}[1]{}%
\renewcommand{\emphind}[1]{\emph{#1}\index{#1}}

\definecolor{blue25emph}{rgb}{0, 0, 11}

\providecommand{\emphic}[2]{}
\renewcommand{\emphic}[2]{\textcolor{blue25emph}{%
      \textbf{\emph{#1}}}\index{#2}}

\providecommand{\emphi}[1]{}%
\renewcommand{\emphi}[1]{\emphic{#1}{#1}}

\definecolor{almostblack}{rgb}{0, 0, 0.3}

\providecommand{\emphw}[1]{}%
\renewcommand{\emphw}[1]{{\textcolor{almostblack}{\emph{#1}}}}%

\providecommand{\emphOnly}[1]{}%
\renewcommand{\emphOnly}[1]{\emph{\textcolor{blue25}{\textbf{#1}}}}

\newcommand{\myqedsymbol}{\rule{2mm}{2mm}}
\newcommand{\SarielThanks}[1]{%
   \thanks{%
      Department of Computer Science; %
      University of Illinois; %
      201 N. Goodwin Avenue; %
      Urbana, IL, 61801, USA; %
      \href{mailto:spam@illinois.edu}{sariel@illinois.edu}; %
      \url{http://sarielhp.org/}. %
   #1%
   }%
}

\newcommand{\HLink}[2]{\hyperref[#2]{#1~\ref*{#2}}}
\newcommand{\HLinkSuffix}[3]{\hyperref[#2]{#1\ref*{#2}{#3}}}

\newcommand{\figlab}[1]{\label{fig:#1}}
\newcommand{\figref}[1]{\HLink{Figure}{fig:#1}}

\providecommand{\deflab}[1]{\label{def:#1}}
\newcommand{\defref}[1]{\HLink{Definition}{def:#1}}
\newcommand{\defrefY}[2]{\hyperref[def:#2]{#1}}

\newcommand{\apndlab}[1]{\label{apnd:#1}}
\newcommand{\apndref}[1]{\HLink{Appendix}{apnd:#1}}

\newcommand{\seclab}[1]{\label{sec:#1}}
\newcommand{\secref}[1]{\HLink{Section}{sec:#1}}

\newcommand{\exmlab}[1]{\label{example:#1}}
\newcommand{\exmref}[1]{\HLink{Example}{example:#1}}

\newcommand{\lemlab}[1]{\label{lemma:#1}}
\newcommand{\lemref}[1]{\HLink{Lemma}{lemma:#1}}%

\providecommand{\eqlab}[1]{}%
\renewcommand{\eqlab}[1]{\label{equation:#1}}
\newcommand{\Eqref}[1]{\HLinkSuffix{Eq.~(}{equation:#1}{)}}

\providecommand{\remove}[1]{}%
\newcommand{\Set}[2]{\left\{ #1 \;\middle\vert\; #2 \right\}}

\newcommand{\pth}[1]{\mleft(#1\mright)}%

\newcommand{\ProbC}{{\mathbb{P}}}
\newcommand{\ExC}{{\mathbb{E}}}

\newcommand{\MedianC}{\mathbb{M}}

\newcommand{\ExCond}[2]{\ExC\!\left[%
       #1 \;\middle\vert\; #2 \right]}
\newcommand{\MedCond}[2]{\MedianC\!\left[%
       #1 \;\middle\vert\; #2 \right]}

\newcommand{\Prob}[1]{\ProbC\mleft[ #1 \mright]}
\newcommand{\ProbCond}[2]{\mathop{\ProbC}\!\left[%
       #1 \;\middle\vert\; #2 \right]}
\newcommand{\Ex}[1]{\ExC\mleft[ #1 \mright]}

\newcommand{\Median}[1]{\MedianC\mleft[ #1 \mright]}

\newcommand{\ceil}[1]{\mleft\lceil {#1} \mright\rceil}
\newcommand{\floor}[1]{\mleft\lfloor {#1} \mright\rfloor}

\renewcommand{\th}{th\xspace}

\renewcommand{\Re}{\mathbb{R}}%

\newcommand{\profileX}[1]{\ensuremath{\mathrm{profile}\pth{#1}}}

\usepackage{todonotes}

\newlist{compactenumA}{enumerate}{5}%
\setlist[compactenumA]{topsep=0pt,itemsep=-1ex,partopsep=1ex,parsep=1ex,%
   label=(\Alph*)}%

\newlist{compactenuma}{enumerate}{5}%
\setlist[compactenuma]{topsep=0pt,itemsep=-1ex,partopsep=1ex,parsep=1ex,%
   label=(\alph*)}%

\newlist{compactenumI}{enumerate}{5}%
\setlist[compactenumI]{topsep=0pt,itemsep=-1ex,partopsep=1ex,parsep=1ex,%
   label=(\Roman*)}%

\newlist{compactenumi}{enumerate}{5}%
\setlist[compactenumi]{topsep=0pt,itemsep=-1ex,partopsep=1ex,parsep=1ex,%
   label=(\roman*)}%

\newlist{compactitem}{itemize}{5}%
\setlist[compactitem]{topsep=0pt,itemsep=-1ex,partopsep=1ex,parsep=1ex,%
   label=\ensuremath{\bullet}}%

\usepackage{color}

\newcommand{\xbeginlgox}{\begin{minipage}{1in}\begin{tabbing}
           \quad\=\qquad\=\qquad\=\qquad\=\qquad\=\qquad\=\qquad\=\kill}
        \newcommand{\xendlgox}{\end{tabbing}\end{minipage}}

\newenvironment{program}{
   \begin{minipage}{4.0in}
   \begin{tabbing}
       \ \ \ \ \= \ \ \ \= \ \ \ \ \= \ \ \ \ \= \ \ \ \ \=
      \ \ \ \ \= \ \ \ \ \= \ \ \ \ \= \ \ \ \ \=
      \ \ \ \ \= \ \ \ \ \= \ \ \ \ \= \ \ \ \ \= \kill
}{
   \end{tabbing}
   \end{minipage}
}

\DefineNamedColor{named}{AlgorithmColor}{cmyk}{0.07,0.90,0,0.34}

\numberwithin{figure}{section}%
\numberwithin{table}{section}%
\numberwithin{equation}{section}%

\DeclareFontFamily{U}{BOONDOX-calo}{\skewchar\font=45 }
\DeclareFontShape{U}{BOONDOX-calo}{m}{n}{
  <-> s*[1.05] BOONDOX-r-calo}{}
\DeclareFontShape{U}{BOONDOX-calo}{b}{n}{
  <-> s*[1.05] BOONDOX-b-calo}{}
\DeclareMathAlphabet{\mathcalb}{U}{BOONDOX-calo}{m}{n}
\SetMathAlphabet{\mathcalb}{bold}{U}{BOONDOX-calo}{b}{n}
\DeclareMathAlphabet{\mathbcalb}{U}{BOONDOX-calo}{b}{n}

\usepackage{stmaryrd}%

\newcommand*{\T}{\mathcal{T}}

\newcommand{\Opt}{\mathcal{O}}

\newcommand{\NN}{\mathbb{N}}%
\newcommand{\RExt}{\mathbb{R}^{+\infty}}%

\newcommand{\Distrib}{\mathcal{D}}
\newcommand{\alg}{\texttt{alg}\xspace}

\newcommand{\hprofile}{\mathcalb{h}}
\newcommand{\hprofileX}[1]{\hprofile\pth{#1}}%

\newcommand{\medianC}{\mathcalb{m}}%
\newcommand{\medianY}[2]{\medianC_{#1}( #2)}%
\newcommand{\medianX}[1]{\medianC(#1)}%

\newcommand{\pr}{\mathcalb{p}}

\newcommand{\etal}{\textit{et~al.}\xspace}

\newcommand{\TTL}{\textsf{TTL}\xspace}

\newcommand{\proxy}{\mathcalb{t}}%
\newcommand{\EWX}[1]{\mathcalb{w}_{#1}}%
\newcommand{\fTTLX}[1]{\mathcal{T}_{#1}}

\newcommand{\bitsX}[1]{\mathcalb{b}\pth{#1}}
\newcommand{\BIN}{\ensuremath{\textbf{BIN}}\xspace}
\newcommand{\Good}{\mathcal{G}}%

\newcommand{\Valis}{\texttt{L16}\xspace}
\newcommand{\Cluster}{\texttt{L128}\xspace}

\newcommand{\BugTrap}{\texttt{BugTrap}\xspace}

\newcommand{\BugTrapShort}{\texttt{BugTrap}$_{0.25}$\xspace}

\newcommand{\BugTrapLong}{\texttt{BugTrap}$_{2.0}$\xspace}

\newcommand{\Maze}{\texttt{Maze}\xspace}

\newcommand{\invX}[1]{\alpha^{-1}\pth{#1}}%
\newcommand{\Seq}{\mathcal{T}}%

\BibLatexMode{%
   \bibliography{theory}
}

\begin{document}

\title{Quickly Avoiding a Random Catastrophe}

\author{%
   Stav Ashur%
   \and%
   Sariel Har-Peled\SarielThanks{Work on this paper was partially
      supported by NSF AF award CCF-2317241.  }}

\date{\today}

\maketitle

\begin{abstract}
    We study the problem of constructing simulations of a given
    randomized search algorithm \alg with expected running time
    $O( \Opt \log \Opt)$, where $\Opt$ is the optimal expected running
    time of any such simulation. Counterintuitively, these simulators
    can be dramatically faster than the original algorithm in getting
    \alg to perform a single successful run, and this is done without
    any knowledge about \alg, its running time distribution, etc.

    For example, consider an algorithm that randomly picks some
    integer $t$ according to some distribution over the integers, and
    runs for $t$ seconds. then with probability $1/2$ it stops, or
    else runs forever (i.e., a catastrophe). The simulators
    described here, for this case, all terminate in constant expected
    time, with exponentially decaying distribution on
    the running time of the simulation.

    Luby \etal \cite{lsz-oslva-93} studied this problem before -- and
    our main contribution is in offering several additional simulation
    strategies to the one they describe.  In particular, one of our
    (optimal) simulation strategies is strikingly simple: Randomly
    pick an integer $t>0$ with probability $c/t^2$ (with
    $c= 6/\pi^2$). Run the algorithm for $t$ seconds. If the run of
    \alg terminates before this threshold is met, the simulation
    succeeded and it exits. Otherwise, the simulator repeat the
    process till success.
\end{abstract}

\section{Introduction}

Consider a randomized algorithm \alg, with running time $X$ (in
seconds), where $X$ is a random variable. If the algorithm stops
naturally, it had completed a \emphi{successful} run.  Such randomized
algorithms that always succeed but their running time is a random
variable are known as \emphw{Las Vegas} algorithms.  Consider the
worst case scenario -- all we know is that the probability \alg is
successful is bounded a way from zero.  In practice, a failure might
be a case where the algorithm just takes longer than what is
acceptable (e.g., days instead of seconds).
\begin{example}
    \exmlab{1:100}%
    Consider an algorithm \alg with its running time $X$, such that
    $\Prob{X =1} = 0.01$ (i.e., \alg terminates after one second with
    probability $1/100$. Otherwise, it runs forever.
\end{example}
Our purpose is to simulate \alg such that the simulation has a small
(expected) running time, and has a ``lighter'' tail (i.e., one can
prove a concentration bound on the running time of the simulation).

\paragraph{The stop and restart model.}
We assume, that initializing the algorithm is ``free'', and we can
stop the algorithm's execution if the runtime exceeds a certain
prespecified threshold known as \emphi{\TTL} (\emphi{time to
   live}).  In the simplest settings, all we can do is restart the
algorithm and run it from scratch (using fresh randomization, thus the
new run is independent of earlier runs) once the \TTL is exceeded.

\begin{example}
    Consider the \alg described in \exmref{1:100}, with
    $\Prob{X=1}=1/100$, and otherwise it runs forever.  The optimal
    simulation strategy here is self evident -- run \alg for a \TTL of
    one second, if it terminates, viola, a successful run was
    found. Otherwise, terminate it and start a fresh run with the
    same \TTL. Clearly, the running time of this simulation has
    geometric distribution $\mathrm{Geom}(1/100)$, with expected
    running time of a $100$ seconds. A dramatic improvement over the
    original algorithm, which most likely would have never terminated.
\end{example}

Of course, in most cases the distribution $\Distrib$ of the running
time of $\alg$ is unknown. The challenge is thus to perform efficient
simulation of \alg without knowing $\Distrib$.  This problem was
studied by Luby \etal \cite{lsz-oslva-93}, who presented an elegant
simulation scheme that runs in (expected) $O( \Opt \log \Opt)$ time,
where $\Opt$ is the expected running time of the optimal simulation
scheme for the given algorithm. They also prove that no faster
simulation is possible.

\paragraph{Time scale.}

For concreteness, we use seconds as our ``atomic'' time units, but the
discussion implies interchangeability with other units of time.  In
particular, we consider one second as the minimal running time of
$\alg$.

\paragraph{Simulation.}

Here, we consider various simulation schemes of \alg, where the
purpose is to minimize the total running time till one gets a
successful run. We consider several different simulation models,
starting with a simple ``stop and restart'' model, and reaching more involved
schemes where the simulation creates several copies of the original
algorithm and run them in parallel, or dovetail their execution by
pausing and resuming their execution.

\paragraph{Relevant literature.}

The work of Luby \etal \cite{lsz-oslva-93} (mentioned above) studied
the problem we revisit here.  It was inspired in turn by the
work\footnote{Earlier, despite the years in the two references.}  of
Alt \etal \cite{agmkw-moras-96}, which provided a simulation strategy
to minimize the tail of the distribution of the simulator running
time.

Luby and Ertel \cite{le-oplva-94} studied the parallel version of Luby
\etal \cite{lsz-oslva-93}. They point out that computing the optimal
strategy in this case seems to be difficult, but they prove that under
certain assumptions a fixed threshold strategy is still
optimal. Furthermore, they show that running the strategy of Luby
\etal~on each CPU/thread yields theoretically optimal threshold.
Truchet \etal \cite{trc-ppsul-13} studied the running time in practice
and in theory of the algorithm if one simply run the algorithm in
parallel. They seems to be unaware of previous relevant work.

The simulation techniques of Luby \etal are used in solvers for
satisfiability \cite{gkss-ss-08}, and constraint programming
\cite{rbw-hcp-06}.  They are also used in more general search
problems, such as path planning \cite{maq-acopp-23}.

More widely, such simulations are related to stochastic resetting
\cite{ers-sra-20}, first passage under restart \cite{pr-fpr-17}, and
diffusion with resetting \cite{tpsrr-erdsr-20}. Over simplifying,
these can be described as random processes when one restart the
process till the system reaches a desired state.

\paragraph{The pause and resume model.}

An alternative, more flexible, model allows us to pause an algorithm's
execution if its running time exceeds a certain threshold, resume
its execution for a prespecified time, and repeat this process for as
many times as necessary.  Thus, one can run several (independent)
copies of the original algorithm in ``parallel'' by dovetailing the
iterations between them (for example, by round-robin scheduling the
different copies of the algorithm one by one). One can show that any
scheme for this model, can be simulated (with a loss of constant
factor).  by the stop and restart model.

\paragraph{Simulation is for free.}

We ignore the overhead of running the simulation, as this seems to be
negligible compared to the total running time of the algorithm being
simulated.

\paragraph{Our contribution.}

We first describe several known simulation strategies and their
performance. This is done to familiarize the reader with the settings,
and get some intuition on the problem at hand. These strategies
include (a $\star$ indicates optimal strategy):
\begin{compactenumI}
    \medskip%
    \item \TTL: If full information about the running time distribution
    of \alg is provided, then the optimal strategy is always to
    rerun the algorithm with the same \TTL. This was shown by Luby
    \etal and we rederive it  in \secref{TTL:best}, for the sake of
    completeness.

    In particular, the optimal threshold $\proxy$ induces a
    \emph{profile} $(k, \proxy) \in \Re^2$ that fully describes the
    behavior of \alg with respect to the \TTL simulation -- specifically,
    $k$ is the (expected) number of times one has to run \alg with the
    optimal threshold $\proxy$ before a successful run is
    achieved. Importantly, any simulation of \alg under any strategy
    must take (in expectation) $\Theta( k \proxy)$ time.  The concept
    of a profile and its properties are presented in \secref{profile},
    and it provides us with a clean way to think about simulation
    strategies for \alg.

    \medskip%
    \item \textsf{Exponential search}.  In \secref{exp:search} we
    explore maybe the most natural strategy -- exponential search for
    the optimal threshold \TTL. It turns out to be a terrible,
    horrible, no good strategy in the worst case (but we describe
    cases where it is decent, and it is maybe the first strategy to
    try out in practice).

    \medskip%
    \item \textsf{Parallel search}. The most lazy approach in practice
    is to run as many copies of \alg as there are threads in the
    system, and stop as soon as one of the copies succeeds.  if one
    runs $k$ copies of \alg, the expected running time of the
    simulation behaves like the $k$\th order statistic of the original
    distribution. While this ``bending'' of the running time of \alg
    does result in a speedup, it is easy to verify that it is not
    competitive in the worst case.

    \medskip%
    \item \textsf{Counter search} $(\star)$. We next present the
    optimal simulation strategy of Luby \etal \cite{lsz-oslva-93}, and
    prove it indeed takes $O( \Opt \log \Opt)$ expected time.  Its
    name is due to the \TTL{}s to use at every step being fully
    specified by tracking the bits that change when increasing the
    value of a binary counter by $1$. See \secref{counter} for
    details.

    \medskip%
    \item \textsf{$\zeta(2)$ search} $(\star)$. In \secref{r:zeta:2}
    we describe a random strategy that is strikingly simple: Randomly
    pick an integer $t>0$ with probability $c/t^2$ (with
    $c= 6/\pi^2$).  Use this as the \TTL for a run, and repeat the
    process till success.  With high probability this strategy is
    optimal, see \lemref{zeta:2:h:p}.

    \medskip%
    \item \textsf{Random search} $(\star)$.  In \secref{r:counter} we
    describe another random strategy that tries to simulate the
    ``distribution'' over the integers formed by the counter
    search. Informally, the algorithm generates a random binary string
    of length $2^k$ (with most significant bit $1$) with probability
    $1/2^{k}$. All strings of length $k$ have the same probability to
    be output. We then use the number encoded by the string as the
    \TTL threshold for the run. As usual, this repeats till a
    successful run is found.
    This strategy is also optimal, see     \lemref{r:counter}.

    \medskip%
    \item \textsf{Wide search} $(\star)$.  In \secref{wide:s}, we
    suggest another optimal simulation strategy. The basic idea is to
    run many copies of \alg, but in different ``speeds'', where the
    $i$\th copy is run at speed $1/i$. This can be achieved by
    suspending/resuming a run of \alg as needed.

    \medskip%
    \item \textsf{X+Cache search} $(\star)$.  Maybe the most
    interesting idea to come out of the wide search strategy is the
    notion of suspending a process. A natural strategy then is to run
    any of the other optimal strategies, but whenever a run exceeds
    its \TTL, instead of killing it, we suspend it. Next time a run is
    needed with a higher \TTL, we resume the suspended processes and
    run it for the remaining desired time. One needs to specify here
    the size of the cache -- i.e., how many suspended processes can
    one have (as too many suspended processes in practice seems to
    lead to a degradation in performance, as happens with the wide
    search). If the cache is full a process is killed as in the
    original scheme.  This is described in detail in
    \secref{cache}. This strategy is optimal if $X$ is one of the
    above optimal strategies.

\end{compactenumI}

\remove{%
\paragraph{Paper organization.}

In \secref{basic}, we establish the basic results on this problem, and
review known results.  In \secref{TTL:best}, we show that \TTL
strategy is an optimal under full information -- this part is from
Luby \etal and is provided for the sake of completeness.  In
\secref{profile}, we introduce the notion of profile of an algorithm.
If \alg has profile $(k, \Delta)$, then somewhat informally, \alg
needs to be run $k$ times, with a \TTL of $\Delta$ to get a successful
run with constant probability (the notion of profile might be knew).
In \secref{simulations} we describe and analyze two strategies: (i)
exponential search (which is natural but suboptimal), and (ii) counter
search (due to Luby \etal).

In \secref{new} we describe our new simulation strategies. We start at
\secref{random} with random search, where the \TTL{}s being used are
randomly sampled, and we offer two distributions: (i) the $\zeta(2)$
distribution in \secref{r:zeta:2}, and (ii) binary strings
distribution \secref{r:counter} (which simulates the counter search).
In \secref{wide:s}, we describe the wide search -- a strategy that
pauses and resumes copies of \alg instead of killing them. At least
theoretically it should lead to a (small) constant speedup over the
other strategies, as it is more efficient.  Finally, in
\secref{cache}, we point out that one can combine any of the previous
strategies (except the wide search) with a simple caching strategy
where some of the short term processes are being killed, while longer
ones are suspended and resumed using a ``cache'' of suspended
processes.  }

\section{Basic simulation schemes}%
\seclab{basic}

\paragraph*{Settings and basic definitions.}

Let $\RExt$ denote the set of all real non-negative numbers including
$+\infty$.  Let \alg denote the given randomized algorithm --- we
assume its running time is distributed according to an (unknown)
distribution $\Distrib$ over $\RExt$. Given a threshold $\alpha$, and
a random variable $X \sim \Distrib$, the \emphi{$\alpha$-median} of
$\Distrib$, denoted by
$\medianY{\alpha}{X}= \medianY{\alpha}{\Distrib}$, is the minimum
(formally, infimum) of $t$ such that
\begin{math}
    \Prob{X < t} \leq \alpha \text{ and } \Prob{X > t} \leq 1-\alpha.
\end{math}
The \emphi{median} of $\Distrib$, denoted by $\medianX{\Distrib}$, is
the $1/2$-median of $\Distrib$.

\begin{defn}
    \emphi{\TTL} simulation. We are given a prespecified threshold
    $\Delta$, \alg is run with $\Delta$ as \TTL threshold. Once the
    running time of this execution exceeds the provided \TTL $\Delta$
    without succeeding, this copy of \alg is killed and a new copy is
    started using the same \TTL $\Delta$. This continues until a copy
    of \alg stops naturally before its execution time exceeds its
    \TTL.
\end{defn}

\paragraph{Specifying a simulation strategy.}
The basic idea is to generate a sequence (potentially implicit and
infinite) of \TTL{}s $\Seq = t_1, t_2, \ldots$, and run \alg using
these \TTL{}s one after the other, till a run ends successfully, or
its time exceeds its \TTL, and then it is killed and a new run
starts. The question is, of course, how to generate a sequence that
minimizes the expected time till success is encountered. Thus, the
strategy of the \TTL simulation with threshold $\Delta$, is the
sequence $\Seq = \Delta, \Delta, \ldots$.

\subsection{The optimal strategy is always \TTL simulation}
\seclab{TTL:best}

Assume we are explicitly given the distribution of the running time of
the algorithm \alg. For concreteness (and simplicity), assume we are
given a discrete distribution --- with times
$s_1 \leq s_2\leq \cdots\leq s_n$, and corresponding probabilities
$p_1, \ldots, p_n$, where $p_i$ is the probability the algorithm stops
after exactly $s_i$ seconds.

Let $F_i = \sum_{j=1}^i p_j$. Clearly, an optimal strategy would
always set the time threshold to be one of the values of
$S = \{ s_1, s_2, \ldots, s_n \}$. If the algorithm selects threshold
$s_i$ then if this iteration was successful the expected time
this final iteration took is
\begin{math}
    \mu_i = \tfrac{1}{F_i}\sum_{\ell=1}^i p_\ell s_\ell.
\end{math}
For the strategy that always picks $t_i$ as a threshold, its expected
number of iterations till success is $1/F_i$, where the last one is a
success iteration for which in expectation we have to only pay
$\mu_i$ (instead of $s_i$).

Let $\Opt$ denotes the expected running time till success under the
optimal strategy. As before, if we decide the first threshold should
be $s_i$, then we have
\begin{equation*}
    \Opt = F_i \mu_i + (1-F_i) (s_i +  \Opt)
    \implies
    F_i \Opt = F_i \mu_i + (1-F_i) s_i
    \implies
    \Opt =  \mu_i + \frac{1-F_i}{F_i} s_i .
\end{equation*}

\begin{lemma}[Full information settings \cite{lsz-oslva-93}]
    \lemlab{full:k}%
    Let \alg be a randomized algorithm with running time being a
    random variable $X \in \RExt$. The expected running time of \alg
    using \TTL simulation with threshold $t$ is
    \begin{equation*}
        f(t)
        =%
        \ExCond{X}{X \leq t} + \frac{\Prob{X > t }}{\Prob{X\leq t}} t
        \leq%
        R(t)
        \qquad
        \text{where }\qquad%
        R(t) =
        \frac{t}{\Prob{X\leq t}}.
    \end{equation*}
    The \emphi{optimal threshold} of \alg is
    \begin{math}
        \Delta = \fTTLX{\alg} = \arg\min_{t \geq 0} f(t).
    \end{math}
    The optimal simulation of \alg, minimizing the expected running
    time, is the \TTL simulation of \alg using threshold $\Delta$, and
    its expected running time is $\Opt = f(\Delta)$.
\end{lemma}

\begin{proof}
    The above readily implies the stated bounds. What might be less
    clear is why applying the same threshold repeatedly
    is the optimal strategy.  To this end, observe that the
    process is memoryless -- the failure of the first round (which is
    independent of later rounds) does not change the optimal strategy
    to be used in the remaining rounds\footnote{This is where we use
       the full information assumption -- we gain no new information
       from knowing that the first run had failed.}. Thus,
    the threshold to be used in the second round should also be used
    in the first round, and vice versa. Thus, all rounds should use
    the same threshold.
\end{proof}

\begin{example}[Useful]
    \exmlab{zeta:2}%
    The Riemann zeta unction is $\zeta(s) = \sum_{i=1}^\infty 1/i^s$.
    In particular, $\zeta(2) = \pi^2/6$ (i.e.,
    \href{https://en.wikipedia.org/wiki/Basel_problem}{Basel
       problem}).  In particular, consider the scenario that the
    running time of \alg is a random variable $X \sim \zeta(2)$. That
    is, $\Prob{X=i} = c/i^2$, where $c = 6/\pi^2$.  This is a special
    case of the
    \href{https://en.wikipedia.org/wiki/Zeta_distribution}{Zeta
       distribution}, specifically \emphi{$\zeta(2)$
       distribution}. The expected running time of \alg is unbounded
    as $\Ex{X} = \sum_{i=1}^\infty i c/i^2 = +\infty$. It is not hard
    to verify that the minimum of $f$ is realized at $\Delta=1$, and
    the expected running time of this \TTL simulation is $O(1)$.
\end{example}

\begin{example}[Not useful]
    Consider the scenario that the running time is a random variable
    $X$ that is uniformly distributed in an interval
    $ I = [\alpha,\beta]$, for $\beta \geq \alpha \geq 0$. It is not
    hard, but tedious, to verify that the optimal threshold is
    $\beta$, which is the same as running the original algorithm till
    termination.
\end{example}

\subsection{Algorithms and their profile}
\seclab{profile}

If $\alg$ has optimal threshold $\Delta = \fTTLX{\alg}$ , and
$\beta =\Prob{ X\leq \Delta}$. The optimal strategy would rerun $\alg$
(in expectation) $1/\beta$ times, and overall would have expected
running time bounded by $\Delta / \beta$. Since the functions $f(t)$
and $R(t)$ of \lemref{full:k} are not the same, we need to be more
careful.
\begin{defn}
    \deflab{profile}%
    For the algorithm \alg, let $X$ be its running time.  The
    \emphw{proxy running time} of \alg with threshold $\alpha$ is the
    quantity $R(\alpha) = \alpha/ \Prob{X \leq \alpha}$. Let
    $\proxy  = \arg \min_{\alpha > 0 } R(\alpha)$ be the
    optimal choice for the proxy running time.
    The \emphi{profile} of $\alg$, is the point
    \begin{equation*}
        \profileX{\alg}
        =
        (1/\pr, \proxy), \qquad\text{where}\qquad
        \pr = \Prob{X \leq \proxy}.
    \end{equation*}
    The total expected \emphi{work} associated with $\alg$ is
    $\EWX{\alg} = R(\proxy) = \proxy / \pr$.
\end{defn}

We show below that the optimal simulation strategy, up a constant, for
\alg with profile $(1/\pr, \proxy)$ is to run the \TTL simulation with
threshold $\proxy$. In expectation, this requires running \alg $1/\pr$
times, each for $\proxy$ time. Thus, the work associated with $\alg$
is the area of the axis aligned rectangle with corners at the origin
and $\profileX{\alg}$. Thus, all the algorithms with the same amount
of work $\alpha$ have their profile points on the hyperbola
$\hprofile(\alpha) \equiv( y=\alpha/x)$.

If \alg has
profile $(1/\pr, \Delta)$ then the expected running time
of the simulation is determined (roughly) by the minimum $k$, such that
$t_1, \ldots, t_k$ contains at least $1/\pr$ numbers at least as large
as $\Delta$. In particular, the expected simulation running time is
(roughly) $O(\sum_{i=1}^k t_i)$.

Intuitively (but somewhat incorrectly) $f(x)$ and $R(x)$ are roughly
the same quantity, and one can thus consider $R(x)$ to be the expected
running time of the \TTL simulation with threshold $x$. The following
lemma shows that this is indeed the case for the minimums of the two
functions, which is what we care about. The proof of this lemma is a
bit of a sideshow, and is delegated to the appendix.

\begin{lemma}[Proof in \apndref{proof:equiv}]
    \lemlab{equiv}%
    If \alg has profile
    $(1/\pr, \proxy)$, then the optimal simulation of \alg takes in
    expectation $\Theta(R(\proxy)) = \Theta( \proxy/\pr)$ time.
\end{lemma}

\subsection{Simulations using \TTL sequences}
\seclab{simulations}

We now get back to task of designing good simulation strategies.  A
natural way to figure out all the profiles of algorithms for which the
prefix $t_1, \ldots, t_k$ (of a sequence $\Seq$) suffices in
expectation, i.e., \alg should have a successful run using only these
\TTL{}s, is to sort them in decreasing order and interpret the
resulting numbers $t_1' \geq t'_2, \ldots, t_k'$ as a bar graph (each
bar having width $1$). All the profiles under this graph, are
``satisfied'' in expectation using this prefix. We refer to the
function formed by the top of this histogram, as the $k$-\emphi{front}
of $\Seq$.  see \figref{work}.

\begin{figure}[h]
    \phantom{}\hfill%
    \includegraphics[width=0.45\linewidth]{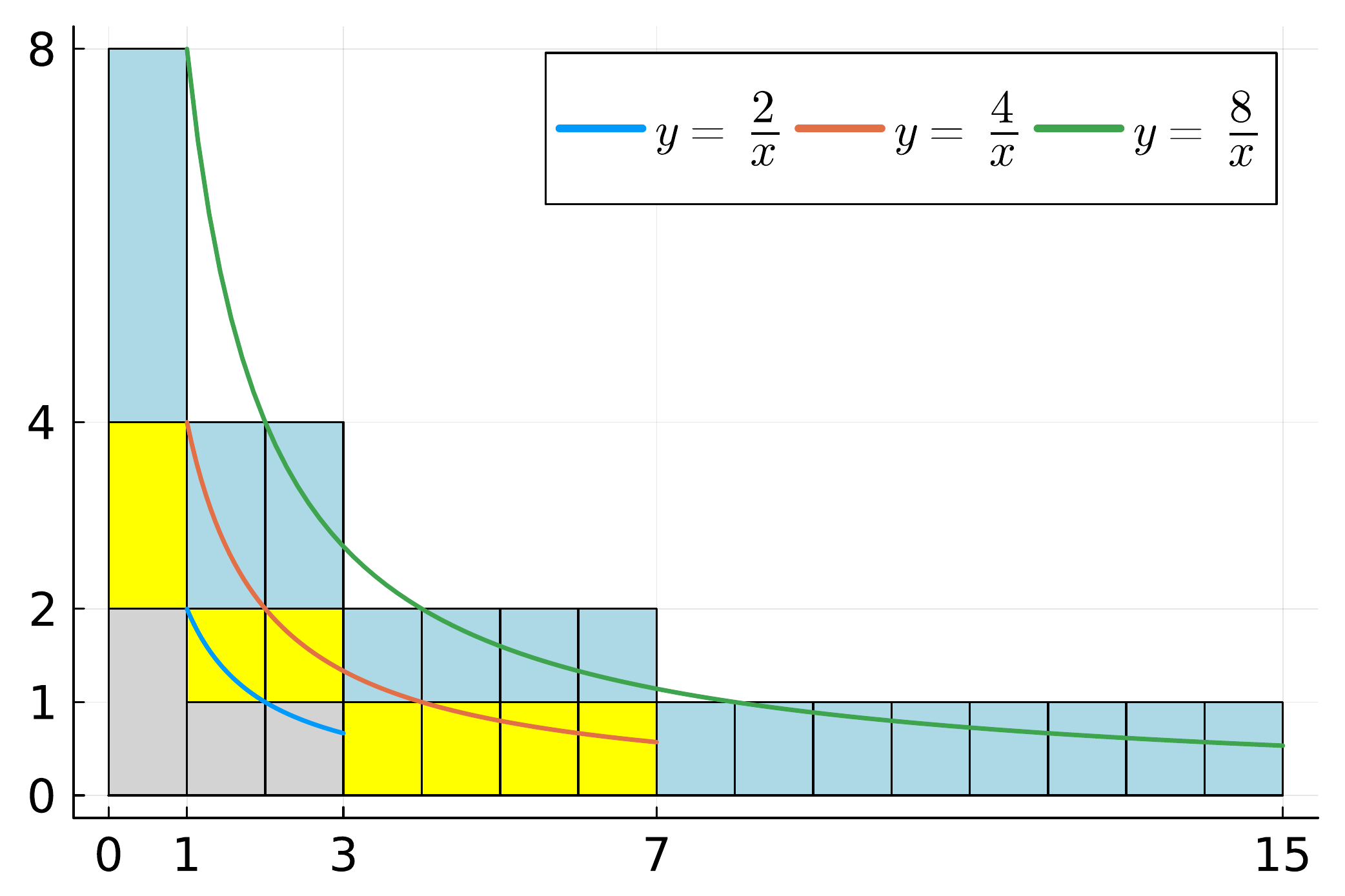}
    \hfill%
    \includegraphics[width=0.45\linewidth]{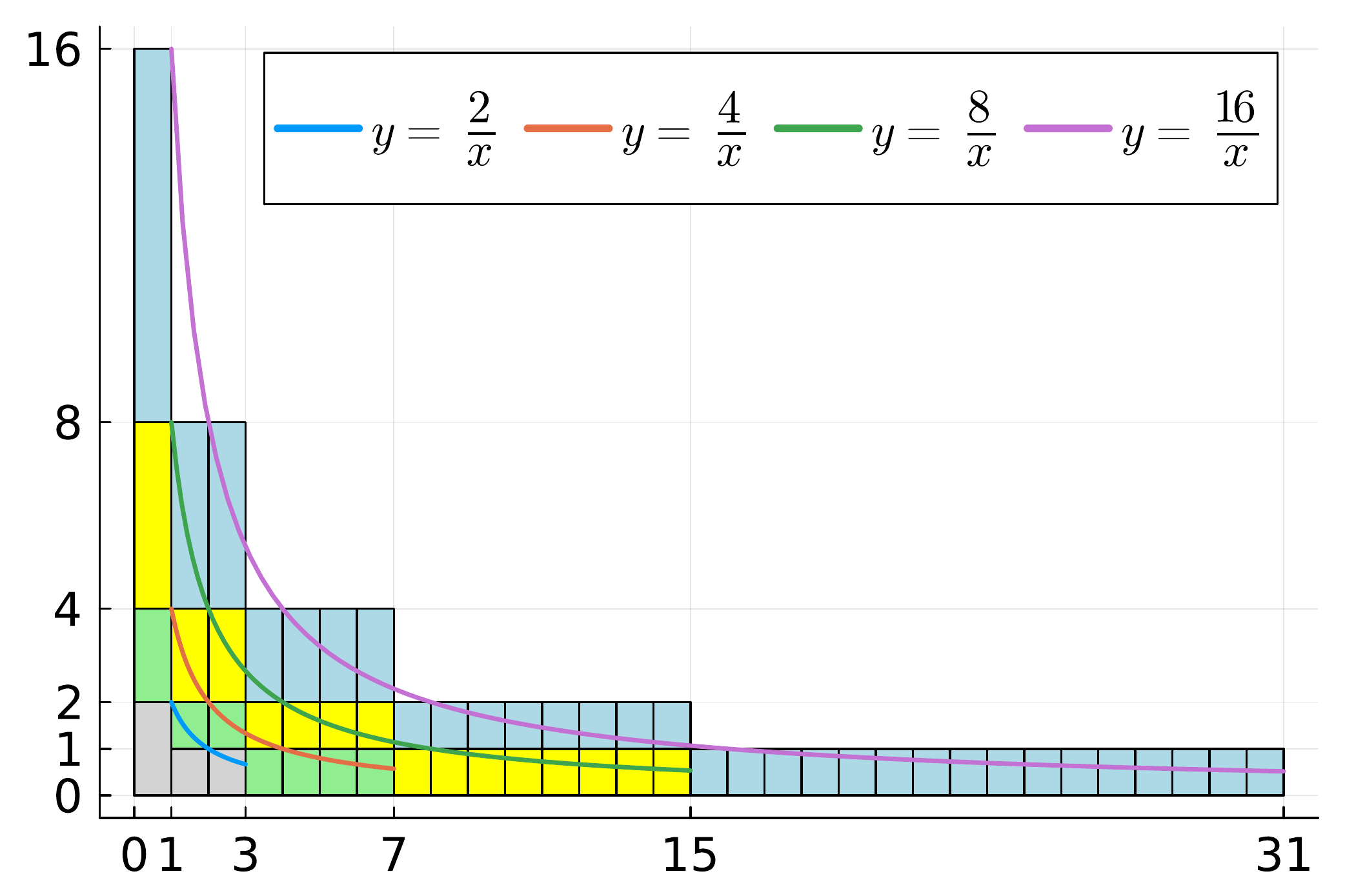}
    \hfill%
    \phantom{}%
    \caption{The work hyperbola for $8$ (left) and $16$ (right), and
       their associated sequences $S_3$ and $S_4$.}
    \figlab{work}
\end{figure}

A key insight is that even if we knew that the profile of \alg lies on
the hyperbola $\hprofile(\alpha) \equiv( y=\alpha/x)$ (as mentioned
above, all these algorithms have the same expected work $\alpha$), we
do not know where it lies on this hyperbola. Thus, an optimal sequence
front should sweep over the \emphi{$\alpha$-work hyperbola}
\begin{equation*}
    \hprofileX{\alpha}
    =%
    \Set{ \Bigl. (1/\pr, t) }{ t/\pr = \alpha \text{ and } 1 \leq 1/\pr
       \leq \alpha }
\end{equation*}
more or less in the same time. The question of how to come up with
such a strategy is not obvious, and we start with the natural, but
wrong, solution.

\subsubsection{Exponential search: Simple but not optimal}
\seclab{exp:search}

Consider the situation where the distribution $\Distrib$ of the
running time of \alg is unknown.  One can try to estimate the optimal
threshold $\Delta$, but the challenge is to simulate \alg directly,
without ``wasting'' time on learning $\Distrib$ first.  Given a
parameters $s,\delta \geq 1$, the \emphw{$(s,\delta)$-exponential
   simulation}, uses the \TTL sequence $\Seq = r_1, r_2, \ldots$,
where $r_i = (1+\delta)^{i-1} s$ for all $i \geq 1$.  The proof of the
following is straightforward and is delegated to the appendix.

\begin{lemma}[Proof in \apndref{proof:exp:search}]
    \lemlab{exp:search}%
    (A) Let $\medianC = \Median{\Distrib}$ be the median running time
    of $\alg$.  The $(s,\delta)$-exponential simulation, assuming
    $s \leq \medianC$ and $\delta \leq 1/2$, has expected running time
    $O\bigl(\medianC(1 + 1 / \delta) \bigr)$.

    (B) Exponential search can be arbitrarily slower than the optimal
    simulation.
\end{lemma}

\subsubsection{Counter search}
\seclab{counter}

Here we describe a strategy due to Luby \etal \cite{lsz-oslva-93} --
our exposition is different and is included here for the sake
completeness.

In the following, let $\{x\}^y$ be the sequence made out of $y$ copies
of $x$.  Fix a value $2^i$. We are trying to ``approximate'' the
hyperbola $\hprofile(2^i)$ with a sequence whose front roughly
coincides with this hyperbola.  One way to do so, is to consider a
sequence $S_i$ containing the numbers
\begin{equation*}
    S_i \qquad\equiv\qquad \{2^i\}^{2^0}, \{2^{i-1}\}^{2^1}\ldots \{2^{0}\}^{2^i}
    =%
    \{2^i\}^{1}, \{2^{i-1}\}^{2},  \{2^{i-2}\}^{4},
    \ldots ,
    \{2\}^{2^{i-1}},
    \{1\}^{2^i}
\end{equation*}
(not necessarily in this order!).  See \figref{work}.

The intuition for this choice is the following -- we might as well
round the optimal threshold of the algorithm to the closest (bigger)
power of $2$, and let $\nabla$ denote this number. Unfortunately, even
if we knew that the work of the optimal simulation is $\alpha = 2^i$,
we do not know which value of $\nabla$ to use, so we try all possible
choices for this value. It is not hard to verify that $S_i$ dominates
the work hyperbola $\hprofile{2^i}$ (for all integral points).

What is not clear is how to generate an infinite sequence that its
prefix is $S_i$ if one reads the first $|S_i|$ elements of this
sequence, for all $i$. As Luby \etal \cite{lsz-oslva-93} suggest,
ignoring the internal ordering of $S_i$, we have that
$S_{i+1} = S_i, S_i, 2^{i+1}$, where $S_0 = 1$. Indeed, every number
in $S_i$ appears twice as many times in $S_{i+1}$ except for the largest value,
which is unique. We have $S_1 = 1,1,2$, $S_2 = 1,1,2,1,1,2,4$, and so
on.

\newcommand{\TX}[1]{\texttt{#1}}

Algorithmically, there is a neat way to generate this \emphi{counter
   simulation} \TTL sequence.  Let $c$ be an integer counter
initialized to $0$. To generate the next elements in the sequence, at
any point in time, we increase the value of $c$ by one, and add all
the numbers $2^i$, for $i=0,1,2,3,\ldots$, that divide $c$ into the
sequence.  Interpreting $c$ as being represented in base $2$, this is
equivalent to outputting all the bits have changed by the increment.
Thus, the sequence $T$ is the following.
\begin{equation*}
    \begin{array}{|l||c|c|c|c|c|c}
      \hline
      T\equiv
      &
        1,
      &
        1,2,
      &
        1,
      &
        1,2,4,
      &
        1
      &
        \cdots
      \\
      \hline
      \hline
      c
      &
        1
      &
        2
      &
        3
      &
        4
      &
        5
      &
        \cdots
      \\
      \hline
      c \text{ in base }2
      &
        00\underline{1}
      &
        0\underline{10}
      &
        01\underline{1}
      &
        \underline{100}
      &
        10\underline{1}\Bigr.
      &
        \cdots
      \\
      \hline
    \end{array}
\end{equation*}

\remove{%
\begin{equation*}
    T \equiv
    \underbrace{1}_{c=1:~\TX{00\underline{1}}},%
    \quad%
    \underbrace{1,2}_{c=2:~\TX{0\underline{10}}},%
    \quad%
    \underbrace{1}_{c=3:~\TX{01\underline{1}}},%
    \quad%
    \underbrace{1,2,4}_{c=4:~\TX{\underline{100}}},
    \quad
    \underbrace{1}_{c=5:~\TX{10\underline{1}}}, \ldots
\end{equation*}
}

\begin{lemma}[\cite{lsz-oslva-93}]
    Running the stop and restart simulation on an algorithm \alg
    using the sequence $T$ above results in expected running time
    $O( \Opt \log \Opt)$, where $\Opt$ is the expected running time of
    the optimal simulation for \alg.
\end{lemma}
\begin{proof}
    Let $S(i,k)$ be the sum of all the appearances of $2^i$ in first
    $k$ elements of $T_k$. For $2^i > 2^j$ (both appearing in the
    first $k$ elements), we have $S(i,k) \leq S(j,k) \leq 2S(i,k)$, as
    one can easily verify.

    Let $(1/\beta, \Delta)$ be the profile of $\alg$, see
    \defref{profile}.  The optimal simulation takes in expectation
    $\Opt= \Theta(\Delta/\beta)$ time.  Let $k$ be the number of
    elements of $T$ that were used by the simulation before it
    stopped.  Let $B$ be the number of elements in $t_1, \ldots, t_k$ that are
    at least as large as $\Delta$.  Clearly,
    $B \sim \mathrm{Geom}(\beta)$ (to be more precise it is
    stochastically dominated by it). Thus, we have
    $\Ex{B} \leq 1/\beta$. Furthermore, the quantity
    $\ell = 1 +\ceil{\log_2 (B \Delta)}$ is an upper bound on the
    largest index $\ell$ such that $2^\ell$ appears in
    $t_1, \ldots, t_k$.  An easy calculation shows that
    $\Ex{\ell} =O(1 + \log (\Delta/\beta) ) = O( \log \EWX{\alg})$.
    The total work of the simulation thus is
    $\sum_{i=0}^\ell S(i,k) = O( \ell B \Delta)$, which in expectation
    is $O( \EWX{\alg} \log \EWX{\alg} ) = O( \Opt \log \Opt)$.
\end{proof}

\section{New strategies}
\seclab{new}

\subsection{The random search}
\seclab{random}

A natural approach is to generate the sequence of \TTL{}s using some
prespecified distribution, repeatedly picking a value $t_i$ before
starting the $i$\th run. This generates a strategy that can easily be
implemented in a distributed parallel fashion without synchronization,
with the additional benefit of being shockingly simple to describe.

\subsubsection{The Zeta $2$ distribution}
\seclab{r:zeta:2}

\begin{lemma}
    \lemlab{zeta:2:h:p}%
    With probability $\geq 1 -1 /\Opt^{O(1)}$, the expected running
    time of a simulation of $\alg$, using a \TTL sequence sampled from
    the $\zeta(2)$ distribution of \exmref{zeta:2}, is
    $O( \Opt \log \Opt)$, where $\Opt$ is the minimal expected running
    time of any simulation of \alg.
\end{lemma}
\begin{proof}
    Let $(1/\beta, \Delta)$ be the profile of \alg, and let
    $L = c \cdot \Opt^{c}$, where $\Opt = \Theta(\Delta/\beta)$ and
    $c$ is some sufficiently large constant, and let
    $\Seq_L = t_1, t_2, \ldots, t_L \sim \zeta(2)$ be the random
    sequence of \TTL{}s used by the simulation.  In the following, let
    $X \sim \zeta(2)$.  Clearly, the probability that any of the
    values in $\Seq_L$ exceed $L^2$ is at most
    $L \Prob{X \geq L^2} = O(1/L)$ since for $i > 1$ we have
    \begin{equation*}
        \Prob{X \geq i}
        =%
        \sum_{\ell=i}^{\infty} \frac{6/\pi^2}{\ell^2}
        \leq%
        \frac{6}{\pi^2}\sum_{\ell=i}^{\infty} \frac{1}{\ell(\ell -1)}
        =%
        \frac{6}{\pi^2}\sum_{\ell=i}^{\infty} \pth{ \frac{1}{\ell-1}
           -\frac{1}{\ell}}
        =
        \frac{6}{\pi^2(i-1)}.
    \end{equation*}
    A similar argument shows that
    \begin{math}
        \Prob{X \geq i} \geq \frac{6}{\pi^2i}.
    \end{math}
    The event that all the values of $\Seq_L$ are smaller than $L^2$
    is denoted by $\Good$.  Observe that
    \begin{equation*}
        \tau  =\ExCond{X }{X \leq L^2}
        =%
        \frac{1}{\Prob{X \leq L^2}} \sum_{i=1}^{L^2}
        \frac{6i}{\pi^2i^2}
        =
        O( \log L)
        =
        O( \log \Opt).
    \end{equation*}
    Since $\Prob{t_i \geq \Delta} \geq 6/(\pi^2 \Delta)$, the
    probability that the algorithm terminates in the $i$\th run (using
    \TTL $t_i$) is at least $ \alpha = 6 \beta/(\pi^2 \Delta)$. It
    follows that the simulation in expectation has to perform
    $M = 1/\alpha = (6/\pi^2)\Opt$ rounds. Thus, the expected running
    time of the simulation is
    $\sum_{i=1}^M \ExCond{t_i}{\Good} = M \tau = O( \Opt \log \Opt)$.
\end{proof}

\subsubsection{Simulating the counter search distribution}
\seclab{r:counter}%

Consider a generated threshold as a binary string
$b_1 b_2 \ldots, b_k$, where the numerical value of this string, is
the value it encodes in base two: $t = \sum_{i=1}^k b_i 2^{k-i}$.

\begin{defn}
    For a natural number $t > 0$, the number of bits in its binary
    representation is its \emphi{length} -- that is,
    $\bitsX{t} = 1+\floor{\log_2 t}$.
\end{defn}

We generate a random string as follows. Let $s_1 =1$. Next, in each
step the string is finalized with probability $1/2$. Otherwise, a
random bit is appended to the binary string, where the bit is randomly
chosen between $0$ or $1$ with equal probability. Let \BIN the
resulting distribution on the natural numbers.  Observe that for
$R \sim \BIN$, the probability that $R$ length is $k$ (i.e.,
$\bitsX{R} = k$) is $1/2^k$, for $k\geq 1$. Also, observe that $R$ has
uniform distribution over all the numbers of the same length.

\newcommand{\RandomCnt}{Random Counter\xspace}

We refer to the random simulation strategy using $\BIN$ to sampling
the \TTL{}s as the \emphi{\RandomCnt} simulation.

\begin{lemma}
    \lemlab{r:counter}%
    With probability $\geq 1 -1 /\Opt^{O(1)}$, the expected running
    time of the \RandomCnt simulation is $O( \Opt \log \Opt)$.
\end{lemma}
\begin{proof}
    Let $(1/\beta, \Delta)$ be the profile of \alg,
    $L = c \cdot \Opt^{c}$, where $\Opt = \Delta/\beta$ and $c$ is
    some sufficiently large constant, and let
    $T_L = t_1, t_2, \ldots, t_L$ be the random sequence used by the
    simulation. A value $t_i$ of $T_L$ is \emphi{bad} if
    $t_i \geq L^2$. The probability of $T_L$ containing any bad value
    is at most $L/L^2$, which is polynomially small in $1/\Opt$. From
    this point on, we assume all the values of $T_L$ are good, and let
    $\Good$ denote this event.

    For $R \sim \BIN$, and any integer number $t$, we have
    \begin{equation}
        \frac{2}{t}
        \geq
        \frac{1}{2^{\bitsX{t} -1}}
        =
        \Prob{\bitsX{R} \geq \bitsX{t}}
        \geq
        \Prob{R \geq t}
        \geq
        \Prob{\bitsX{R} > \bitsX{t}}
        =%
        \frac{1}{2^{\bitsX{t}}}
        \geq
        \frac{1}{t}.
        \eqlab{in:between}
    \end{equation}
    Observe that %
    \begin{math}
        \ExCond{R}{\bitsX{R} =j}
        =%
        (2^j + 2^{j+1}-1)/2
        \leq
        3 \cdot 2^{j-1}.
    \end{math}
    Furthermore, $\Prob{R \geq 2^i} = 1/2^{i-1}$, and for $i > j$, we
    have $\ProbCond{\bitsX{R} =j }{ R < 2^i} =
    2^{-j}/(1-1/2^{i-1})$. Thus, we have
    \begin{equation*}
        \ExCond{R}{R < 2^i}
        =%
        \sum_{j=1}^{i-1}
        \ExCond{R}{\bitsX{R} =j}
        \ProbCond{\bitsX{R} =j }{ R < 2^i}
        \leq %
        \sum_{j=1}^{i-1}
        \frac{3 \cdot 2^{j-1} }{2^j(1-1/2^{i-1})}
        \leq
        3i.
    \end{equation*}
    In particular, for $i=1,\ldots, L$, we have
    $\Ex{t_i} = O( \log \Opt )$ conditioned on $\Good$.

    A value $t_i$ of $T_L$ is \emphw{final} if $t_i \geq \Delta$, and
    the running of the $i$\th copy of the algorithm with \TTL $t_i$
    succeeded. Let $X_i$ be an indicator variable for $t_i$ being
    final. Observe that
    \begin{equation*}
        p
        =
        \Prob{X_i=1}
        =
        \ProbCond{ \alg \text{ successful run with \TTL ~}t_i }{t_i \geq
           \Delta}
        \Prob{t_i \geq \Delta}
        \geq%
        \beta \cdot \frac{1}{\Delta}
        =
        \frac{1}{\Opt},
    \end{equation*}
    by \Eqref{in:between}.  The number of values of $T_L$ the
    simulation needs to use till it hits a success is a geometric
    variable $\sim \mathrm{Geom}(1/\Opt)$. (Clearly, the probability
    the algorithm reds all the values of $T_L$ is diminishingly small,
    and can be ignored.) Thus, the algorithm in expectation reads
    $ 1/p = O(\Opt)$ values of $T_L$, and each such value has
    expectation $O( \log \Opt)$. Thus, conditioned on $\Good$, the
    expected running of the simulation is $O( \Opt \log \Opt )$.
\end{proof}

\subsection{The wide search}
\seclab{wide:s}

Consider running many copies of the algorithm in parallel, but in
different speeds.  Specifically, the $i$\th copy of \alg, denoted by
$A_i$, is run in speed $\alpha_i=1/i$, for $i=1,2,3,\ldots$. So, consider
the simulation immediately after it ran $A_1$ for $t$ seconds so far
--- here $t$ is the \emphi{current time} of the simulation.  Thus, at
time $t$ the $i$\th algorithm $A_i$ run so far
$n_i(t) = \floor{ t/i}$ seconds. Naturally, the simulation instantiates
an algorithm $A_i$ at time $i$ since $n_i(i) = 1$, and then runs it
for one second. More generally, the simulation allocates a running
time of one second to a copy of the algorithm $A_j$ at time $t$ such that $n_j(t)$ had increased to a new integral value. Then, $A_j$ is
being ran for one second. This simulation scheme can easily be
implemented using a heap, and we skip the low-level details.

It is not hard to simulate different speeds of algorithms in the
pause/resume model by using a scheduler that allocates time slots for
the copies of the algorithm according to their speed.

\bigskip

The wide search starts to be effective as soon as the running time of
the first algorithm exceeds a threshold $\nabla$ that already provides
a significant probability that the original algorithm would stop with
a success, as testified by the following lemma.

\begin{lemma}
    \lemlab{w:s:good:p}%
    The wide search simulation described above, executed on \alg with
    profile $(1/\pr, \Delta)$, has the expected running time
    $O\pth{ \frac{\nabla}{\beta} \log \frac{\nabla}{\beta}} = O( \Opt
    \log \Opt)$, where $\Opt$ is the optimal expected running time of
    any simulation of \alg (here, the $i$\th algorithm is being ran with
    speed $\alpha_i = 1/i$).
\end{lemma}
\begin{proof}
    Let $T_j = \nabla\ceil{j /\beta} $. If we run the first algorithm
    for $T_j$ time, then the first $\ceil{j/\beta}$ algorithms in the
    wide search are going each to be run at least
    $T_j / \ceil{j/\beta} = \nabla$ time.
    The probability that all
    these algorithms failed to terminate successfully is at most
    $p_j = (1- \beta)^{\ceil{j/ \beta}} \leq \exp(-j)$.
    The total running time of
    the simulation until this point in time is bounded by
    \begin{equation*}
        S_j = \sum_{i=1}^{T_j} \frac{T_j}{i} = O(T_j \log T_j)
        =%
        O\Bigl( \frac{\nabla j}{\beta} \log \frac{\nabla j}{\beta} \Bigr).
    \end{equation*}
    Thus, the expected running time of the simulation is
    asymptotically bounded by
    \begin{equation*}
        \sum_{j=0}^\infty p_j O(S_{j+1})
        =
        \sum_{j=0}^\infty
        O\pth{ \exp(-j)  \frac{(j+1)\nabla}{\beta}
           \log \frac{(j+1)\nabla}{\beta}
        }
        = %
        O\pth{ \frac{\nabla}{\beta} \log
           \frac{\nabla}{\beta}}
        =%
        O( \Opt \log \Opt ),
    \end{equation*}
    as $\Opt = O( \nabla/\beta)$.
\end{proof}

For additional results on the wide search, see \apndref{more:wide:s}.

\subsection{Caching runs}
\seclab{cache}

A natural approach suggested by the wide simulation, is to run any of the
other simulations based on a sequence of \TTL{}s, but instead of
terminating a running algorithm when its \TTL is met,
we suspended it instead. Now, whenever the algorithm needs to be run for a
certain time $t_i$, we first check if there is a suspended
run that can be resumed to achieve the desired threshold. To avoid
spamming the system with suspended runs (as the wide search does), we
limit the cache of suspended jobs, and kill jobs if there is no space
for them in the cache.

This recycling loses some randomness that the original scheme has, but
it has the advantage of being faster, as long runs of \alg are not
started from scratch.

\section{Conclusions}

We revisited a paper of Luby \etal \cite{lsz-oslva-93} studying the
problem of how to simulate a Las Vegas algorithm optimally. We
described in details their results, and presented several new
strategies and proved that they are also optimal. In a companion
paper, we carried out extensive testing of these various
strategies. This problem is quite interesting from both theoretical
and practical points of view, and it surprising that it is not better
known in the randomized algorithms literature/community

\paragraph{The benefit of parallelism.}

A major consideration in practice that one can use many
threads/processes in parallel (i.e., we used a computer with 128
threads in our experimental study.

\BibTexMode{%
   \bibliographystyle{alpha}%
   \bibliography{shortcuts,theory}%
}%
\BibLatexMode{\printbibliography}

\appendix

\section{More on the wide search}
\apndlab{more:wide:s}

Here, we present some additional results on the wide search, in
particular showing how it leads to an optimal strategy in the
stop/restart model (i.e., similar to the counter search). We show that
in some cases one can get better expected running time using a
different strategy by choosing different ``speeds'' for the
simulation.

\subsection{Simulation in the restart model}

One can convert a wide-search simulation to a stop/restart simulation.
Indeed, consider the pause/resume wide-search simulation -- and
consider how it handles the $i$\th copy of the algorithm $A_i$. Every
once in a while, the simulator resumes this algorithm and runs its for
one second. Assume that the simulator already ran algorithm $A_i$
for $t$ seconds, and wants to resume it. The simulator can simply
start from scratch and run a new copy of the algorithm for $2t$
seconds, where the new copy replaces $A_i$, and it is dormant till the
simulation tries to run $A_i$ pass $2t$ seconds. (This is the standard
idea used in implementing $O(1)$ time insertion into an array.)  Since
the runs performed for a single algorithm $A_i$ are a geometric
series, the running time is dominated by the last execution for it,
which implies that the total running time of the simulation
is proportional to its overall running time. We thus get the following
result.

\begin{theorem}
    Let \alg be a randomized algorithm, and assume that an optimal
    simulation for \alg is able to run \alg successfully in expected
    $\Opt$ time. Then, the simulation described above, in the stop and
    restart model, takes $O( \Opt \log \Opt )$ expected time.
\end{theorem}

\subsection{Improving the expected running time in some cases}

Here, we show that in some cases one can do better than
$O( \Opt \log \Opt)$ expected running time.

\begin{assumption}
    One can choose other sequences of $\alpha$s (i.e., specify
    different speeds for the algorithms being simulated).  The natural
    requirement is that the sequence of $\alpha$s is monotonically
    decreasing, and furthermore, for any $i \in \NN$, we have that
    $\alpha_{2i} \geq c \alpha_i$, for some absolute constant
    $i$. This implies that the $\alpha$s decrease polynomially with
    $i$ (not exponentially). Specifically, there exists some constant
    $d$, such that $\alpha_i = \Omega( 1/i^d)$.
\end{assumption}

Given such a sequence of $\alpha$s, let
$\invX{t} = \min \Set{i \in \NN}{ 1/\alpha_i > t}$. The total work
done by the simulation by time $t$ is
\begin{equation*}
    W(t)
    =
    \sum_{i=1}^{\invX{t}} \floor{t\alpha_i}.
\end{equation*}
We emphasize that if any of these algorithms succeeds, then the
simulation stops\footnote{And declares victory, hmm, success.}.

The first time when the simulation ran at least $k$ algorithms for $t$
time, is when $n_k(\tau) = t$, which happens when $\tau = t/\alpha_k$
(as $A_1, \ldots, A_{k-1}$ all ran for longer by this time). The total
work done by the simulation by this time is $W(\tau) = W(t/\alpha_k)$.

\begin{lemma}
    \lemlab{expected}%
    Let $\mu$ be the expected running time of \alg. Then, the expected
    running time of the simulation is $O(W(\mu))$, and furthermore,
    for $\psi \geq 1$, the probability that the simulation runs for
    more than $\Omega(W(2 \mu/\alpha_\psi) )$ time, is at most
    $2^{-\psi}$.
\end{lemma}
\begin{proof}
    The time when the simulation run (at least) $i$ algorithms (each)
    for more than $2 \mu$ time, is $\tau_i = 2\mu /\alpha_i$.  By
    Markov's inequality, the probability of $A_j$ to run for more than
    $\alpha_j \tau_i = 2\mu(\alpha_j/\alpha_i)$ time, before stopping
    is at most
    \begin{equation*}
        \frac{\mu}{2\mu(\alpha_j/\alpha_i)}
        =%
        \frac{ \alpha_i }{2 \alpha_j} \leq \frac{1}{2}.
    \end{equation*}
    Thus, the probability the simulation did not stop by time $\tau_i$
    is at most
    \begin{equation*}
        p_i
        =%
        \prod_{k=1}^i \frac{ \alpha_i }{2 \alpha_k}
        \leq
        \frac{1}{2^i}.
    \end{equation*}
    Thus, since $\alpha_i \geq c/i^d$, for some constants $c>0$ and
    $d>1$, the expected running time of the simulation is
    asymptotically bounded by
    \begin{equation*}
        \sum_{i=1}^\infty p_{i-1} W(2\mu/\alpha_i)
        \leq%
        \sum_{i=1}^\infty \frac{1}{2^i}
        \sum_{\ell=1}^{2 \mu /\alpha_i} \alpha_\ell \frac{2 \mu}{\alpha_i}
        \leq%
        \sum_{i=1}^\infty \frac{1}{2^i}
        \sum_{\ell=1}^{2 \mu /\alpha_i}  \frac{2 \mu i^d }{c}
        =
        \sum_{i=1}^\infty O( \mu i^d/2^i)
        =
        O(\mu).
    \end{equation*}
\end{proof}

Thus, for any ``reasonable'' polynomially decaying sequence
$\alpha_1 \geq \alpha_2 \geq \cdots$, the expected running time of the
resulting simulation is
$O( \mu)$. By choosing the sequence more carefully, one can make the
tail probability smaller.

\begin{lemma}
    \lemlab{wide:search:i}%
    Consider an instantiation of the above simulation, where the
    $i$\th algorithm runs at speed $\alpha_i =1/(i \log^2 i)$, for
    $i>2$, and $\alpha_1 = 1$, $\alpha_2 = 1/2$.  Then the expected
    running time of the simulation is $O(\mu)$. Furthermore, the
    probability simulation runs for more than $O(u \mu)$ time, for any
    $u \geq 10$, is at most $1/\exp\pth{\Omega( u/ \log^2 u )}$.
\end{lemma}
\begin{proof}
    It is easy to verify that $\sum_{i=1}^\infty \alpha_i = O(1)$, and
    thus, for any $x > 1$, we have $W(x) = O(x)$.  Setting
    $\psi = u/ \log^2 u$, we have by \lemref{expected}, that the
    algorithm takes more than
    $\Omega(W(2 \mu/\alpha_\psi) ) = \Omega( u \mu )$ is at most
    $1/2^{\psi}$, which implies the claim.
\end{proof}

\section{Proofs}

\subsection{Proof of \lemref{equiv}}
\apndlab{proof:equiv}

\begin{proof}
    Let $X$ be the running time of \alg.  The proof in mathematically
    neat, but the reader might benefit from skipping it on a first
    read.

    Using the notations of \lemref{full:k}, let
    $\Delta = \fTTLX{\alg}$, and let $O( f(\Delta) )$ be the expected
    running time of the optimal simulation of $\alg$. Let
    $\beta = \Prob{X \leq \Delta}$.  \lemref{full:k} implies that the
    expected running time of the \TTL simulation of \alg with \TTL
    $\Delta$ is
    \begin{equation*}
        f(\Delta)
        =%
        \ExCond{X}{X \leq \Delta} + \frac{1- \beta}{\beta} \Delta
        \leq
        \frac{\Delta}{\beta},
    \end{equation*}
    and $f(\Delta) \leq f(\proxy) \leq R(\proxy) = \proxy/\pr$. So
    $O(R(\proxy))$ is an upper bound on the expected running time of
    the optimal simulation of \alg.

    If $\beta \leq \tfrac78$, then
    $f(\Delta) \geq \tfrac18 \Delta/\beta = \tfrac18 R(\Delta) \geq
    \tfrac18 R(\proxy) = \Omega(\proxy/\pr)$.

    So assume $\beta > \tfrac78$. Let
    $\medianC = \MedCond{X }{ X \leq \Delta}$ and
    $\xi = \ExCond{X}{X \leq \Delta} $.  By Markov's inequality, we
    have that $\medianC \leq 2 \xi$. Namely,
    $f(\Delta) \geq \medianC/2$.  Observe that
    \begin{equation*}
        \Prob{ X \leq \medianC }
        =
        \Prob{ (X \leq \medianC) \cap (X \leq \Delta)  }
        =
        \ProbCond{ X \leq \medianC}{X \leq \Delta  } \Prob{X \leq \Delta}
        \geq
        \frac{ \beta}{2}
        >
        \frac{7}{16}.
    \end{equation*}
    Thus, we have
    \begin{math}
        R(\proxy)%
        \leq%
        R(\medianC)%
        =%
        \frac{\medianC}{\Prob{X \leq \medianC}}%
        \leq%
        \frac{16}{7}\medianC%
        \leq%
        5 f(\Delta).
    \end{math}
\end{proof}

\subsection{Proof of \lemref{exp:search}}
\apndlab{proof:exp:search}

\begin{proof}
    (A) Let $r_i = (1+\delta)^{i-1}s$, for $i\geq 1$. Let $\alpha$ be
    the first index such that $r_\alpha \geq \medianC$, and observe
    that for $i \geq \alpha$, we have that the probability the $i$\th
    run of the algorithm exceeds its time quota, is at most $1/2$ (as
    the running time in this iteration has probability $1/2$ to be
    below the median). For $R$ being the expected running time of this
    simulation, we have
    \begin{equation*}
        \Ex{R}
        \leq
        \sum_{i=1}^{\alpha-1} s(1+\delta)^{i-1}
        +
        \sum_{i=\alpha}^{\infty} \frac{1}{2^{i-\alpha-1}}
        s(1+\delta)^{i-1}
        \leq
        \frac{2 \medianC}{\delta}
        + \frac{2 \medianC}{1- (1+\delta)/2}
        \leq
        8 \medianC + \frac{2\medianC}{\delta}.
    \end{equation*}

    \noindent%
    (B)
    Unfortunately, the exponential search simulation can be
    arbitrarily bad compared to the optimal.  Let $\alg$ be an
    algorithm that with probability $\alpha$ stops after one second,
    and otherwise runs forever, where $\alpha \in (0,1)$ is some
    (small) parameter. The optimal strategy here is the fixed search,
    with $\Delta=1$ (i.e., rerun \alg for one second till success) ---
    and it has expected running time $1/\alpha$. Now, consider the
    $(1,\delta)$-exponential simulation. Since $\alpha$ is not known
    to the simulation, we can set to any value such that (say)
    $\alpha \leq \min( \delta/2, 1/4)$.
    Observe that
    \begin{math}
        (1-\alpha) (1+2\alpha) = 1 +\alpha - 2\alpha^2 \geq 1+
        \alpha/2.
    \end{math}
    Thus, the expected running time of the $(1,\delta)$-exponential
    simulation is
    \begin{equation*}
        \sum_{i=1}^\infty  (1-\alpha)^{i-1} (1+\delta)^{i-1}
        \geq%
        \sum_{i=1}^\infty  \bigl( (1-\alpha) (1+2\alpha)\bigr)^{i-1}
        \geq
        \sum_{i=1}^\infty  \bigl( 1+\alpha/2) = \infty.
    \end{equation*}
    Namely, the optimal simulation has finite expected running time,
    while the exponential simulation running time is unbounded.
\end{proof}

\remove{%
\section{Bla bla}

    \begin{equation*}
        1= \sum_{\ell=1}^{\infty} \frac{6/\pi^2}{\ell^2}
    \end{equation*}
    \begin{equation*}
        \Delta_i = \sum_{\ell=1}^{i} \frac{6/\pi^2}{\ell^2}
    \end{equation*}

    $[0,\Delta_1]: i=1$

    $[\Delta_1,\Delta_2]: i=2$

    $[\Delta_k,\Delta_{k+1}]: i=k+1$%
 }

\end{document}